\documentclass[a4paper,12pt]{scrartcl}
\usepackage[centertags, intlimits]{amsmath}
\usepackage{amsfonts}
\usepackage{amssymb}
\usepackage{amsthm}
\usepackage{nicefrac}
\usepackage{titlesec}
\usepackage{pifont}
\usepackage[ansinew]{inputenc}
\usepackage{graphicx}
\usepackage{natbib}

\theoremstyle{plain}
\newtheorem{thm}{Theorem}[section]

\newtheorem{rema}{Remark}[section]

\title{On characteristics of an ordinary differential equation
and a related inverse problem in epidemiology}

\author{Ralph Brinks\\
Institute of Biometry and Epidemiology \\German Diabetes Center\\
Duesseldorf, Germany}

\date{}

\begin{document}

\maketitle

\begin{abstract}
In this work we examine the properties of a recently described
ordinary differential equation that relates the age-specific
prevalence of a chronic disease with the incidence and mortalities
of the diseased and healthy persons. The equation has been used to
estimate the incidence from prevalence data, which is an inverse
problem. The ill-posedness of this problem is proven, too.
\end{abstract}


\maketitle

\section{Introduction}
Recently a novel ordinary differential equation (ODE) has been
described that relates the age-specific prevalence of an
irreversible disease with its incidence rate and the mortality
rates of the diseased and the non-diseased persons \citep{Bri11}.
This article is about properties of the ODE and its solutions.

\bigskip

Given the mortality rates and the age-specific prevalence, the ODE
may be used to derive the incidence rate\footnote{In this article
the expressions \emph{rate} and \emph{density} are synonymously
used.}, which can be interpreted as an inverse problem. Inverse
problems are often examined with respect to ill- or
well-posedness. A well-posed of an inverse problem in the sense of
Jacques Hadamard means that a solution exists, that the solution
is unique and stable \citep{Had23}. In this article the
ill-posedness of the inverse problem is proven.

The article is organized as follows. Section \ref{s:methods}
briefly reviews the derivation of the ODE. Then, some properties
of the ODE and its solution are examined. In Section
\ref{s:inverse} the inverse problem is introduced and the
ill-posedness is proven. Finally, in Section \ref{s:Discussion}
the results and its consequences are discussed.

\section{The ODE: Derivation and Properties}\label{s:methods}

A popular framework for studying relations between prevalence and
incidence of a disease is the simple model consisting of three
states as depicted in Figure \ref{fig:3states}: \emph{Normal,
Disease} and \emph{Death}, \citep{Kei91, Mur94}. This model goes
back at least until the 1950s \citep{Fix51}. In general, the
transition densities from one state to another depend on calendar
time $t$ and age $a$. The transition density from \emph{Disease}
to \emph{Death} may also depend on the duration $d$ of the
disease.

\begin{figure*}
\centerline{\includegraphics[keepaspectratio,
width=14cm]{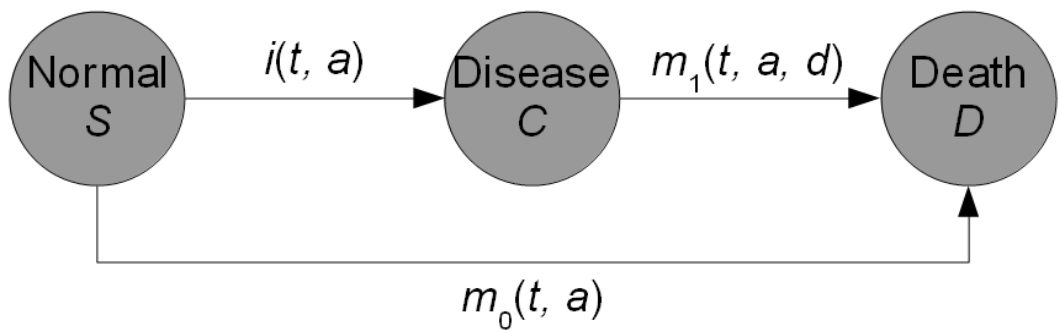}} \caption{Three states model of
normal, diseased and dead subjects. Transition densities may
depend on calender time $t,$ age $a,$ and duration $d$ of the
disease.} \label{fig:3states}
\end{figure*}

People in the population under consideration can get a disease at
incidence density $i$, they can die either after having got the
disease at age-specific mortality rate $m_1$ or without having the
disease at mortality rate $m_0$. The numbers of the individuals in
the \emph{Normal} and in the \emph{Disease} state are denoted by
$S$ (susceptibles) and $C$ (cases). Generally, these numbers as
well as the mortality and incidence densities depend on the
calendar time $t$. Sometimes, these quantities are assumed to be
independent of time $t$, which is referred to as the
time-homogeneous case \citep{Kei91}.

Assuming time-homogeneity and disease related mortality $m_1$ to
be independent of duration $d$, Murray and Lopez
\citeyearpar{Mur94} described transitions along the paths in
Figure \ref{fig:3states} as a set of ODEs. Henceforth, beside
time-homogeneity we additionally assume that the population is
closed (no migration) and has a constant birth rate. Furthermore,
the age-specific functions $i, m_0, m_1$ are non-negative and
continuous in $\left[0, \omega \right]$ for some $\omega > 0$.
Henceforth, $\omega$ is considered as the minimal age when all
members (diseased and non-diseased) of the population are
deceased\footnote{For example by choosing $\omega = \inf \{ a
> 0| ~C(a) + S(a) < 0.5\}$.}. Then, Equation \eqref{eq:MurrayODE}
describes the change rates of the numbers $S$ and $C$ of normal
and diseased individuals, respectively.

\begin{equation}\label{eq:MurrayODE}
\begin{split}
    \frac{\mathrm{d} S}{\mathrm{d} a} &= - \left ( i(a) + m_0(a) \right ) \cdot S\\
    \frac{\mathrm{d} C}{\mathrm{d} a} &= i(a) \cdot S - m_1(a) \cdot C\\
\end{split}
\end{equation}

\bigskip

The resulting set of ODEs is linear and of first order. Due to the
simple structure of the ODEs for given age-specific incidence and
mortality rates $i$, $m_0$ and $m_1$, the analytical solution of
the corresponding initial value problem with initial conditions
$S(0) = S_0 \ge 0, ~C(0) = C_0 \ge 0, ~S_0 + C_0 > 0$ can be
calculated easily:

\begin{equation}\label{eq:SolutionML}
\begin{split}
    S(a) &= S_0 \cdot \exp \left (- \int_0^a i(\tau) + m_0(\tau) \mathrm{d}\tau\right
    ) \\
    C(a) &= \exp \left (- \int_0^a m_1(\tau) \mathrm{d}\tau\right ) \cdot \left
    \{ C_0 + \int_0^a i(\tau) S(\tau) \exp \left ( \int_0^\tau m_1(t) \mathrm{d}t \right
    ) \mathrm{d} \tau \right \}.
\end{split}
\end{equation}

Obviously, from $S_0 + C_0 > 0$ it follows that $S(a) \ge 0, ~C(a)
\ge 0$ and $S(a) + C(a) > 0$ for all $a \in [0, \omega]$.

\bigskip

Usually, for a population under consideration the overall
(general) mortality density $m$ is observed or reported in
life-tables. The mortality density $m$ is a convex combination of
the mortality density $m_0$ of the normals and the mortality
density $m_1$ of the diseased:

\begin{equation}\label{eq:generalMort}
\begin{split}
m(a) &= p(a) \cdot m_1(a) + \left \{ 1 - p \left(a \right) \right
\} \cdot m_0(a) \\ &= m_0(a) \cdot \left \{ p \left(a \right)
\cdot \left( R \left(a \right) - 1 \right) + 1 \right \},
\end{split}
\end{equation}
where $R(a)$ is the relative risk, $R = \tfrac{m_1}{m_0}$. In this
expression $p$ is the prevalence of the disease, which for a
specific age $a$ and $S(a)+C(a) > 0$ can be written as
\begin{equation}\label{eq:agePrev}
p(a) = \frac{C(a)}{S(a)+C(a)}.
\end{equation}

Equation \eqref{eq:generalMort} allows the application of the
ODE-system \eqref{eq:MurrayODE} in the case when $m_0$ and $m_1$
are unknown. In epidemiology this typically is the case. Then the
ODE-system \eqref{eq:MurrayODE} becomes non-linear and does not
have an analytical solution anymore.

\bigskip

Interestingly, the two-dimensional system can be reduced to a
one-dimensional ODE, which is stated in the following Theorem.
This can easily by derived from the quotient rule and Equation
\eqref{eq:MurrayODE}.
\begin{thm}\label{th:central1}
Let mortality densities $m, m_0 \in C^0\left(\left[0, \omega
\right]\right)$ and $S, C \in C^1\left(\left[0, \omega
\right]\right)$ with $S(a) + C(a) > 0$ for all $a \in \left[0,
\omega \right],$ then $p = \frac{C}{S+C}$ is differentiable in
$\left[0, \omega \right ]$ and it holds
\begin{equation}\label{eq:Eq4}
\frac{\mathrm{d} p}{\mathrm{d} a} = (1-p) \cdot \left ( i - \left (m -
m_0 \right ) \right ).
\end{equation}
\end{thm}

Depending on the type of information about the mortality
densities, the ODE \eqref{eq:Eq4} changes its type (Table
\ref{tab:Types}), which is important when solving the ODE. In case
the ODE is linear, an easy analytical solution exists. If the ODE
is of Riccati or Abelian type \citep{Kam83}, a general analytical
solution does not exist. An extensive monograph about Riccati
equations is \citep{Rei72}.

\begin{table*}
 \centering
 \def\~{\hphantom{0}}
 \begin{minipage}{150mm}
  \caption{Types of the ODE \eqref{eq:Eq4} depending on the given mortality.} \label{tab:Types}
  \begin{tabular*}{\textwidth}{@{}l@{\extracolsep{\fill}}l@{\extracolsep{\fill}}l@{\extracolsep{\fill}}}
  \hline
Given mortality  & Right hand side of Eq. \eqref{eq:Eq4} & Type of the ODE\\
\hline
$m, m_0$   & $(1-p) \cdot \left [i - (m - m_0) \right ]$                       & Linear \\
$m, m_1$   & $(1-p) \cdot \left [i - p \cdot \frac{m_1 - m}{1-p} \right ]$     & Linear \\
$m_0, m_1$ & $(1-p) \cdot \left [i - p \cdot (m_1 - m_0)\right ]$              & Riccati\\
$m_0, R$   & $(1-p) \cdot \left [i - p \cdot m_0 \cdot (R - 1) \right ]$       & Riccati\\
$m_1, R$   & $(1-p) \cdot \left [i - p \cdot m_1 \cdot \frac{R - 1}{R}\right ]$& Riccati\\
$m, R$     & $(1-p) \cdot \left [i - m \cdot \frac{p \cdot (R - 1)}{p \cdot (R - 1) + 1}\right]$ & Abelian\\
\hline
\end{tabular*}
\end{minipage}
\vspace*{-6pt}
\end{table*}

\bigskip

The fractions $\tfrac{R - 1}{R}$ and $\tfrac{p \cdot (R - 1)}{p
\cdot (R - 1) + 1}$ in the last two rows, are very well known in
epidemiology. These are the \emph{exposition attributable risk}
(EAR) and the \emph{population attributable risk} (PAR),
respectively, \citep{Woo05}.

\bigskip

Next it is examined, if the solutions of the one-dimensional ODE
\eqref{eq:Eq4} are epidemiologically meaningful, i.e. $p(a) \in
[0, 1]$ for all $a \in [0, \omega].$ For the system
\eqref{eq:MurrayODE} this is clear: the age-specific prevalence
$p(a) = \tfrac{C(a)}{S(a) + C(a)}$ given by the solutions
\eqref{eq:SolutionML} are meaningful for all $a \in [0, \omega]$.
However, it is not obvious that solutions $p$ of \eqref{eq:Eq4}
are between 0 and 1. For the special case that $m_0 = m_1$ -- this
case is called \emph{non-differential mortality} -- it can be
proven directly. Then the solution of \eqref{eq:Eq4} with initial
condition $p(0) = p_0 \in [0, 1]$ is
\begin{equation*}
p(a) = 1 - (1-p_0) \cdot \exp \left (- \int_0^a i(\tau)
\mathrm{d}\tau\right ),
\end{equation*}
and the epidemiological meaningfulness follows immediately. In
case of differential mortality ($m_0 \neq m_1$) epidemiological
meaningfulness cannot not be proven directly, because it has to
include all the different cases of the right hand side of
\eqref{eq:Eq4} in Table \ref{tab:Types}. Instead of a direct proof
we use the correspondence between Equations \eqref{eq:MurrayODE}
and \eqref{eq:Eq4}. Let $N(a) := C(a) + S(a)$ denote the number of
persons alive at age $a$, then it holds $N(a)
> 0$ for all $a \in [0, \omega].$ We augment Equation
\eqref{eq:Eq4} by another ODE in $N$ with $m$ defined in Equation
\eqref{eq:generalMort}:

\begin{equation}\label{eq:Eq4augment}
\begin{split}
\frac{\mathrm{d} p}{\mathrm{d} a} &= (1-p) \cdot \left ( i - \left
(m - m_0 \right ) \right ) \\
\frac{\mathrm{d} N}{\mathrm{d} a} &= - m \cdot N.
\end{split}
\end{equation}

Then we have the following correspondence between the ODE-systems
\eqref{eq:MurrayODE} and \eqref{eq:Eq4augment}:
\begin{thm}\label{th:corr}
\begin{enumerate}
 \item[(A)] If $S(a), C(a)$ are solutions of \eqref{eq:MurrayODE}, then
   $p(a) := \tfrac{C(a)}{S(a) + C(a)}$ and $N(a) := C(a) + S(a)$ are
   solutions of \eqref{eq:Eq4augment}.
 \item[(B)] If $p(a), N(a)$ are solutions of \eqref{eq:Eq4augment}, then
 $S(a) := \{ 1-p(a) \} \cdot N(a)$ and $C(a) := p(a) \cdot N(a)$ are solutions of
 \eqref{eq:MurrayODE}.
\end{enumerate}
\end{thm}
\begin{proof}
This is an easy exercise in calculus.
\end{proof}

From Theorem \ref{th:corr} (B) it follows that $p(a) \in [0, 1]$
for all $a \in [0, \omega]:$ If $p(a)$ is a solution of
\eqref{eq:Eq4augment}, then it has a representation $p(a) =
\tfrac{C(a)}{N(a)}$. Since $0 \le C(a) \le N(a),$ the solution $p$
is epidemiologically meaningful.

\begin{rema}
From Theorem \ref{th:corr} (A) it is obvious that ODE-system
\eqref{eq:MurrayODE} implies $\tfrac{\mathrm{d} N}{\mathrm{d} a} =
- m \cdot N$. This is the defining equation of a \emph{stationary
population}, which is a population with a special type of age
distribution \citep{Pre82}. Hence, \eqref{eq:MurrayODE} is valid
only for stationary populations. Since most populations are
non-stationary, this is a heavy limitation. However, it can be
shown that \eqref{eq:Eq4} holds true in general populations as
long as certain restrictions on the migration rate are fulfilled.
Details are not subject of this work and will be elaborated in a
subsequent paper.
\end{rema}

\section{The Inverse Problem}\label{s:inverse}
A key application for the ODE \eqref{eq:Eq4} is the derivation of
the age-specific incidence rate $i(a)$ from $p(a)$ if the
mortalities (or any equivalent information in the first column of
Table \ref{tab:Types}) are known. In epidemiology, typically
incidences rates are surveyed in follow-up studies, which may be
very lengthy and costly. If the model assumptions for ODE
\eqref{eq:Eq4} hold true, the equation can be solved for $i(a)$.
Beside mortality information, the age course of the prevalence has
to be known, which can be obtained from relatively cheap
cross-sectional studies. An example is shown in \citep{Bri11}.

In such an application with given mortalities, we conclude from an
effect (the prevalence) to the underlying cause (the incidence),
which can be interpreted as an \emph{inverse problem}
\citep{Tar05}. The inverse problem is opposed to the \emph{direct
problem} of inferring from the incidence (i.e. the cause) to the
prevalence (the effect) by ODE \eqref{eq:Eq4}. Now we show that
the inverse problem is ill-posed in the sense of
\citeauthor{Had23} \citeyearpar{Had23}. For given (sufficiently
smooth) mortalities and $p_0 \in [0, 1]$ define the operator $\wp:
C^0([0, \omega]) \rightarrow C^1([0, \omega]), i \mapsto p,$ such
that $p(0) = p_0$ and $p$ is the solution of \eqref{eq:Eq4}. To
show that the inverse problem is ill-posed we prove that
$\wp^{-1}: p \mapsto \tfrac{\mathrm{d}p/\mathrm{d}a}{1-p} + m -
m_0$ is discontinuous. It is sufficient to show this for the
special case of non-differential mortality ($m = m_0$). Let
$C^k([0, \omega]), ~k=0,1,$ be equipped with the maximum norm
$\lVert \cdot \rVert$. Choose $p \in C^1([0, \omega]), ~\epsilon >
0$ and define $p_{\epsilon, n}(a) := p(a) + \epsilon \cdot \sin
(n\cdot a)$. Then, it is $\lVert p - p_{\epsilon, n} \rVert \le
\epsilon$ and
\begin{align*}
\lVert \wp^{-1}(p) - \wp^{-1}(p_{\epsilon, n}) \rVert &= \Bigl
\lVert \frac{\mathrm{d}p}{\mathrm{d}a} \frac{1}{1-p} -
\frac{\mathrm{d}p_{\epsilon, n}}{\mathrm{d}a}
\frac{1}{1-p_{\epsilon, n}} \Bigr \rVert \\
&= \Biggl \lVert  \frac{\frac{\mathrm{d}p}{\mathrm{d}a} \epsilon
\sin(n \cdot)
  + \epsilon \, n \cos(n \cdot) (1-p)}{(1-p)(1-p - \epsilon \sin(n \cdot))} \Biggr
  \rVert.
\end{align*}
For $\epsilon_n := n^{-1/2}$ and $p(a) \neq 1$ the term
$\epsilon_n \, n \cos(n a) (1-p(a))$ is unbounded as $n
\rightarrow \infty,$ which implies that $ \wp^{-1}$ is
discontinuous and the inverse problem is ill-posed.

\section{Discussion} \label{s:Discussion}
By extending the framework of \citep{Mur94} for studying the
relation between prevalence and incidence, it had been found that
prevalence, incidence and mortality are linked by a simple
one-dimensional ODE. In this article it has been shown that the
solutions of this ODE are epidemiologically meaningful. Depending
on the type of mortality information available, the ODE changes
its type, which has implications about existence of general
analytical solutions. In many epidemiologically relevant cases, an
analytical solution does not exist, and numerical treatment has to
used instead.

\bigskip

An important application of the ODE is the derivation of
age-specific incidences from the age distribution of the
prevalence. This article shows that this question can be
interpreted as an inverse problem. Furthermore, the inverse
problem is ill-posed. The proof of the ill-posedness shows that an
additive high frequency distortion ($\epsilon \sin(n \cdot)$) of
the prevalence may lead to an unbounded inaccuracy in the derived
incidence. However, high frequency distortions might be unlikely
in real chronic diseases. Hence, the consequences in practical
epidemiology are unclear so far.

\bigskip

In the discussed ODE model, several assumptions have been made.
The ODE is valid only if incidence and mortality rates are
independent from calendar time. Due to changes in medical
progress, hygiene, nutrition and lifestyle, mortality does undergo
secular trends. Thus, it is appropriate to formulate Equation
\eqref{eq:Eq4} as a partial differential equation, which is
subject of a subsequent paper. Moreover, in real diseases
mortality of the diseased persons depend on the duration of the
disease. An example is diabetes, where the relative mortality over
the diabetes duration is U-shaped \citep{Car08}. Duration
dependency obfuscates the relation between prevalence, incidence
and mortality \citep{Kei91}. Results as easy as presented here are
unlikely not be expected.

A last note refers to the term \emph{chronic disease}: In this
article, \emph{chronic} means \emph{irreversible}, i.e. there is
no way back from the \emph{Disease} to the \emph{Normal} state.
However, most of the results presented here remain true, if there
is remission back to state \emph{Normal}. Then, the fundamental
ODE \eqref{eq:Eq4} has an additional term that depends on the
remission rate \citep{Bri11}.


\begin{thebibliography}{}
\bibitem[\protect\citeauthoryear{Brinks}{2011}]{Bri11}
Brinks, R. (2011). A new method for deriving incidence rates from
prevalence data and its application to dementia in Germany.
\verb"http://arxiv.org/abs/1112.2720v1"

\bibitem[\protect\citeauthoryear{Carstensen et~al.}{2008}]{Car08}
Carstensen, B., Kristensen, J.K., Ottosen, P., and Borch-Johnsen,
K. (2008) The Danish National Diabetes Register: trends in
incidence, prevalence and mortality. {\it Diabetologia} {\bf 51
(12),} 2187--96.

\bibitem[\protect\citeauthoryear{Fix and Neyman}{1951}]{Fix51}
Fix, E. and Neyman, J. (1951). {\it Human Biology} {\bf 23 (3),}
205--241.

\bibitem[\protect\citeauthoryear{Hadamard}{1923}]{Had23}
Hadamard, J. (1923). {\it Lectures on the Cauchy Problem in Linear
Partial Differential Equations} New York: Yale University Press.

\bibitem[\protect\citeauthoryear{Kamke}{1983}]{Kam83}
Kamke, E. (1983). {\it Differentialgleichungen.} Stuttgart:
Teubner.

\bibitem[\protect\citeauthoryear{Keiding}{1991}]{Kei91}
Keiding, N. (1991). Age-specific incidence and prevalence: a
statistical perspective. {\it Journal of the Royal Statistical
Society A} {\bf 154,} 371--412

\bibitem[\protect\citeauthoryear{Murray and Lopez}{1994}]{Mur94}
Murray, C. J. L and Lopez, A. D. (1994). Quantifying disability:
data, methods and results {\it Bulletin of the WHO} {\bf 72 (3),}
481--494

\bibitem[\protect\citeauthoryear{Preston and Coale}{1982}]{Pre82}
Preston, S. H. and Coale, A. J. (1982). Age structure, growth,
attrition, and accession: a new synthesis, {\it Population Index}
{\bf 48 (2)} 217--59

\bibitem[\protect\citeauthoryear{Reid}{1972}]{Rei72}
Reid, W. T. (1972). {\it Riccati Differential Equations} New York:
Academic Press.

\bibitem[\protect\citeauthoryear{Tarantola}{2005}]{Tar05}
Tarantola, A. (2005). {\it Inverse Problem Theory} Philadelphia:
Society for Industrial and Applied Mathematics.

\bibitem[\protect\citeauthoryear{Woodward}{2005}]{Woo05}
Woodward, M. (2005). {\it Epidemiology. Study Design and Data
Analysis} Boca Raton: Chapman \& Hall/CRC.

\end{thebibliography}
\end{document}